\documentclass[11pt]{article}

\usepackage{tikz}
\usepackage{enumerate} 

\usepackage[ruled]{algorithm2e}

\usepackage[colorinlistoftodos]{todonotes}
\usepackage{algpseudocode}% http://ctan.org/pkg/algorithmicx
\usepackage{amssymb,amsmath, amsthm, amstext, graphics, color}
\usepackage{graphicx}

\usepackage{multicol}
\usepackage{color}
\usepackage[noadjust]{cite}
\usepackage{enumitem}
\usepackage{multirow}
\setlist{nolistsep}
\usepackage{hyperref}
\usepackage{comment}
\usepackage{pgf,tikz}
\usepackage{calrsfs}
\usepackage{float}% http://ctan.org/pkg/float
\usetikzlibrary{arrows,shapes}

\newtheorem{theorem}{Theorem}

\newtheorem{remark}{Remark}

\usepackage{authblk}

\begin{document}

\title{A Dynamically Turbo-Charged Greedy Heuristic for Graph Coloring}

\author{Faisal N. Abu-Khzam}
\author{Bachir M. Chahine}

\affil{Department of Computer Science and Mathematics\\
  Lebanese American University\\
  Beirut, Lebanon
}
%\email{faisal.abukhzam@lau.edu.lb}
%\setcounter{Maxaffil}{0}
%\renewcommand\Affilfont{\itshape\small}

\date{\vspace{-5ex}}
\maketitle

\begin{abstract}
We introduce a dynamic version of the graph coloring problem and prove its fixed-parameter tractability with respect to the edit-parameter. This is used to present a {\em turbo-charged} heuristic for the problem that works by combining the turbo-charging technique with
other standard heuristic tools, including greedy coloring.
The recently introduced turbo-charging idea is further enhanced in this paper by introducing a dynamic version of 
the so called {\em moment of regret} and {\em rollback points}.
Experiments comparing the presented turbo-charging 
algorithm to other (non-turbo-charged) heuristics demonstrate its effectiveness.
Our algorithm often produced results that were either exact or
better than all the other available heuristics.
\end{abstract}

\section{Introduction}

A $k$-vertex coloring of an undirected graph $G=(V,E)$ is a mapping from $V$ to $\{1,\cdots k\}$ that assigns different values, or colors, to every pair of adjacent vertices. The minimum value of $k$ for which such a mapping exists is the {\em chromatic number} of $G$, often denoted by $\chi(G)$.
The corresponding Graph Coloring problem (AKA. Chromatic Number problem) consists of finding $\chi(G)$ or, in its search version, a vertex coloring that assigns $\chi(G)$ colors to a given graph $G$.

Graph Coloring is among the most studied graph theoretic problems, perhaps because of the famous Four Color Theorem \cite{fourColoring} as well as its wealth of applications in various domains especially scheduling and frequency assignment problems \cite{DBLP:books/sp/Lewis16}.

On the other hand, a dynamic problem is one whose input is assumed to have changed after some initial ``satisfactory'' solution was found. Formally, a dynamic version of a problem $X$ takes as input a quintuple $(I, I', S, k, r)$ where $I$ and $I'$ are instances of $X$ and $S$ is a solution to $I$, not necessarily optimal. 
The instances $I$ and $I'$ are at a given {\em edit distance} $k$, termed the {\em edit-parameter}. The question posed is whether a solution $S'$ to $I'$, also not necessarily optimal, can be obtained such that the {\em Hamming distance} between $S$ and $S'$ is bounded by $r$, which we refer to henceforth as the {\em increment-parameter}.
This notion of a dynamic problem was introduced by Downey et al. in the context of Dominating Set \cite{downey2014dynamic}. It was preceded (and inspired) by the work of Hartung and Niedermeier on the List Coloring problem  \cite{HartungN13}.

Capitalizing on the above mentioned work and the more recent work in \cite{abu2015parameterized} on dynamic versions of a number of domination and vertex-deletion problems, we study a dynamic version of the Graph Coloring problem (henceforth DGC). In particular, we show that Dynamic Graph Coloring is fixed-parameter tractable with respect to the edit-parameter and para-$NP$-hard with respect to the increment parameter.
Furthermore, we show that our fixed-parameter algorithm for DGC can be used to obtain a turbo-charged heuristic for the Graph Coloring problem. 

In this paper, the Turbo-Charging idea is further improved by introducing a dynamic version of the {\em moment of regret} notion. This also yields a dynamic version of {\em rollback}.
The effectiveness of this approach is evaluated by comparing our experimental results with those of some of the most known graph coloring heuristic algorithms, and for which the published results include tests on the benchmark graphs from the DIMACS Graph Coloring Challenge.
Our experiments show that the turbo-charged heuristic
outperforms known Graph Coloring heuristics and can sometimes achieve results that are close to optimal.

\section{Preliminaries}

Throughout this paper we consider only simple undirected graphs, i.e., no loops and no multiple edges. Given a graph $G = (V,E)$, $V$ and $E$ (or $V(G)$ and $E(G)$) denote the set of vertices and edges in $G$, respectively. The degree, $deg(v)$, of a vertex $v$ is the number of edges incident on $v$. 
%The maximum and minimum degrees of a vertex in $G$ are denoted by $\Delta(G)$ and $\delta(G)$ respectively. 
A vertex cover in a graph $G$ is a subset $S$ of $V(G)$ such that every edge of $G$ has at least one endpoint in $S$. 

The Graph Coloring problem is among the first few classical problems shown to be $NP$-hard in \cite{npHard1979} and \cite{npHard1974}. It is also known to be para-$NP$ hard when parameterized by the number of colors in a target solution \cite{downey1999parameterized}. In other words, when the objective is to color the input graph with at most $k$ colors, the problem is ``hard'' to solve in $O(n^k)$ time.

The current asymptotically-fastest exact algorithm solves Graph Coloring in $O(2^n)$, where $n$ is the number of vertices in the input graph \cite{Bjorklund2009}. %This is interesting merely for theoretical purposes.
There are a number of known heuristic methods that aimed at possibly coping with the problem's computational intractability, for practical purposes. Moreover, there are a few meta-heuristic methods to which our turbo-charging method is not applicable. We will thus refrain from providing an overview of the meta-heuristic approaches but we provide comparison with the most effective tabu search for which experimental results are known on the instances considered in our experiments \cite{simAnn}. 

The most known heuristic method is the greedy algorithm which takes an ordering of the vertices and, for each vertex, assigns the smallest color not yet assigned to a neighboring vertex. An improved method uses a heuristic which changes the ordering of nodes and then uses the same greedy approach.
The most used ordering is Largest Degree First (LF) \cite{gLarFst} and a random selection as tie breaker when two nodes have the same largest degree.

Another known heuristic approach is based on the following observation: if the next greedy coloring step requires the assignment of a new color to a vertex, then attempt to swap colors in the graph in order to free one of them for that vertex and then assign to it the smallest possible color.
This old ``recoloring'' approach appeared in \cite{gInter}. It is known as the Greedy-With-Interchange method.

Finally, the Range Compaction Heuristic \cite{rangecomp2003} is another (more recent) heuristic that can be considered as an iterated-greedy approach. It starts with a proper coloring and iteratively tries to narrow the range of colors used by up-coloring low-valued colors and down-coloring high-valued ones using some randomized limits. This method has some randomization built into it and works with re-coloring. The main objective is a time for quality trade off.

We shall use a simple greedy approach that rebuilds the input graph by adding one edge at a time. Initially one color is assigned to all vertices, then colors are added when needed. The order of edge addition is obviously the key to effective coloring (by which we mean assigning a smallest-possible number of colors). Our choice of this naive method is affected by applicability of the turbo-charging method as described in \cite{Abu-KhzamCESW17, downey2014dynamic}. Despite the simplicity of the base heuristic method, its turbo-charged variant proved to be most effective when compared to the known heuristics. 

\section{Dynamic Graph Coloring}

The parameterized Dynamic Graph Coloring is defined formally as follows.

\vspace{5pt}

\noindent
{\bf Dynamic Graph Coloring (DGC)}\\
\underline{\bf Input:}  A graph $G = (V, E)$ with a $c$-coloring $C$; a graph $G' = (V, E')$ obtained from $G$ with edit-distance 
$d_e(G, G') \leq k$. 
\newline \underline{\bf Parameters:} \space $k$ and/or $r$.
\newline \underline{\bf Question:} \space Does there exist a $(c+r)$-coloring of $G'$?

\vspace{10pt}

An instance of DGC would be a quintuple $(G,G',C,k,r)$ and it is assumed that the sought coloring $C'$ is obtained by at most $r$ re-coloring operations. This parameterized re-coloring approach appears to be similar to the notion of parameterized reconfiguration (see \cite{naomi} for more details) except that we start by an improper coloring of $G'$ and we want to employ at most $r$ (local) changes to obtain a proper coloring. 

\begin{remark}
\label{r-bound}
The increment parameter is often defined as the Hamming distance between the initial solution (for $G$) and the sought solution. This is based on the assumption that any change to the initial solution has a cost. In our case, changing a color or adding a new color would have the same cost. Therefore whenever we have a conflict, instead of solving it by re-coloring the vertices of the graph, it would be better to just add a new color. 
This justifies the way we pose the question in our above formulation of DGC.
\end{remark}

\begin{theorem}
Dynamic Graph Coloring is fixed-parameter tractable with respect to the edit-parameter. 
\end{theorem}

\begin{proof}
Let $(G,G',C,k,r)$ be an instance of Dynamic Graph Coloring. Obviously, at most $2k$ vertices are affected by the $k$ edge additions. According to the previous remark, $r$ would be at most $2k$ which is the maximum number of additional needed colors. 
Observe, however, that at most $k$ of the affected vertices need a new color. In fact, the subgraph $H$ induced by the newly added $k$ edges has a minimal vertex cover of size at most $k$ and by Remark \ref{r-bound}, it would be enough to add colors to the vertices forming a minimum vertex cover since their deletion eliminates all conflicts. Note (again) that we seek to minimize the number of changes made to the coloring $C$. Finding a ``minimum'' vertex cover of size at most $k$ is solvable in $O^{*}(1.2738^k)$ time using the fixed-parameter algorithm of Chen et al. \cite{Chen10}. Of course we would have to apply the latter algorithm at most $\log{k}$ times.
\end{proof}

On the other hand, Dynamic Graph Coloring parameterized only by the increment parameter, $r$, is not likely to fall in the class FPT. 

\begin{theorem}
Dynamic Graph Coloring parameterized by the increment parameter is para-$NP$ hard.
\end{theorem}

\begin{proof}
Since the edit-distance $d_e(G, G')$ is arbitrary in this case, we can reduce the Graph Coloring problem to its dynamic version as follows. 
Let $(G,r)$ be an arbitrary instance of Graph Coloring. We construct an equivalent instance $(G_1,G_2, C, k, r-1)$ of Dynamic Graph Coloring by letting $G_1$ be an edge-less graph on $|V(G)|$ vertices, $G_2 = G$, $C$ is a coloring that assigns a unique color to all vertices of $G_1$, and $k = |E(G)|$. Obviously, $G$ is $r$-colorable if and only if we can find a coloring $C'$ such that $d(C,C')\leq r-1$. The theorem now follows from the fact that, unless $P=NP$, Graph Coloring is not even in the class $XP$ \cite{downey1999parameterized}.
\end{proof}

\section{A Dynamic Turbo-Charging Algorithm}

The {\em turbo-charging} technique introduced by Downey et al in \cite{downey2014dynamic} uses a simple greedy heuristic to apply the fixed-parameter algorithm for the dynamic version of the Dominating Set problem.
In our approach, we will be working on a ``turbo-charged'' greedy heuristic to apply the fixed-parameter algorithm for the dynamic version of Graph Coloring with some additional improvements to the turbo-charging method.

As mentioned earlier, our algorithm starts off with an edge-less graph by assigning color 1 to all the vertices (we may also start with a colored graph but this could possibly reduce the effectiveness of our approach). As we add the edges of the input graph, we add at most one color at a time when necessary.
In fact, edge addition is best performed in a way that avoids color-conflict as much as possible. For this purpose a greedy edge-ordering that tries to minimize conflicts between consecutive edges is employed. 
%In addition to the above parameters, 
%The edge-set of the graph needs to be sorted in a non-conflicting order. This will make the algorithm run faster on large instances.
Before we fully describe the turbo charging part of our algorithm, we introduce two new parameters. 

%First, the {\em dynamic moment of regret}, which is when the solution exceeds the expected result or budget and it is no longer feasible/desired. Usually it is set to a static value (as defined in \cite{downey2014dynamic}), but we henceforth use a new dynamic value which will be dependent on the current state of the algorithm in terms of edge number and chromatic number.

%The {\em dynamic rollback point} is our second enhancement. Usually when a regret moment is reached, a rollback is proceeded until some pre-set point. We introduce a new dynamic value which also takes the state of the algorithm into consideration and the rollback will be according to a dynamically calculated point each time the moment of regret is reached.

\subsection{Dynamic moment of regret.} 

The {\em moment of regret} is reached when the solution exceeds the expected result or budget and it is no longer feasible/desired. Usually it is set to a static value (as defined in \cite{downey2014dynamic}), but we henceforth use a new dynamic value which will be dependent on the current state of the algorithm in terms of edge number and chromatic number.

At each point in time we keep track of the number of edges added to the graph and the number of colors added. Taking the total number of edges remaining to be added to the graph into consideration, a (dynamic) moment of regret parameter is set to a fraction of the number of colors added per ``group of edges'' added. For example, we might be fine with adding 5 colors after adding 100 edges, while adding 3 colors after adding 10 edges might not be acceptable. Therefore our moment of regret is a fraction of the number of colors added per number of edges added.  

Formally, let $c$ be the current total number of assigned colors and let $n_i$  be the number of edges added after assigning color $i$. Let 

\begin{center}

$m = Min_{1\leq i \leq c}\{\frac{n_i}{c - c_i}\}$.

\end{center}

Moreover, let $k$ be the number of colors assigned by the best-known heuristic. Then our moment of regret is reached when $m < |E|/k$.

\subsection{Dynamic rollback point.} 

When a {\em dynamic moment of regret} is reached, we need to rollback to a specific point where the solution was ``acceptable.'' In other words we take back a number of edge additions and we restore the coloring saved up to that point.
To determine this number, we keep track of the variation of the graph coloring function. In particular we keep track of an interval where the function is slowly varying or varying at an acceptable rate. We call this interval the {\em stable interval}. The right endpoint of the stable interval is called the {\em dynamic rollback point} and it is updated also after each application of the turbo-charging subroutine.  
%(i.e., we reached this point after a large number of edges were removed and the graph coloring was slightly or not being modified
Then, instead of undoing a static number of edge additions (as in the original turbo-charging method), we take back edge additions until we reach the dynamic rollback point. 

%which is the (right) endpoint of our stable interval. where the coloring of the graph is not increasing rapidly according to our moment of regret measure (i.e., we reached this point after a large number of edges were removed and the graph coloring was slightly or not being modified, we set this point as our {\em dynamic rollback point} and proceed with edge removal until this point is reached.

%Knowing that the state of the algorithm (number of edges and colors) is saved on each iteration, 
%we do not need to rollback to the start point of the algorithm, contrarily, 
%we need to track back where this solution started to be not feasible and rollback until just before this point.

%We start with the rollback process by proceeding with the edge removal and track back the saved states of our algorithm. 

\subsection{A turbo-charged greedy approach}

Whenever a new edge is added, either the edge does not create any coloring conflict and no new color need to be assigned, or the edge addition creates a conflict and a new color is needed.  
In both cases, the algorithm proceeds in a greedy approach, and adds the least possible color available that does not create a conflict.
Note that our algorithm assumes a new color is needed when it is not possible to change the color of one of the two endpoints. 
%We do not look further to a recoloring of other vertices. This latter strategy might yield an improvement of our algorithm at the cost of increasing the running time.

After each color-addition, the algorithm checks if the defined dynamic regret point is reached.
If not, it proceeds with the next edge addition. 
Otherwise it performs a dynamic rollback option, removing the previously added edges until some dynamic backup point, then the (fixed-parameter) dynamic problem subroutine (DGC\_FPT) is called to determine if a smaller solution can be found.
The subroutine will take both graphs before and after the dynamic rollback process, the current coloring and the current edge index, in addition to the current coloring reduced by one as input parameters in an attempt to reduce the current coloring by at least one color.
A pseudocode of this algorithm is given below.

\vspace{10pt}

\begin{algorithm}[htb!]
    \begin{algorithmic}[1]
        \Procedure{DYN TURBO-Charging}{}   
        %\Comment Graph $G$ 
        %and parameter $k$
      %  \Comment Returns a matching $M$

\State Initialize $C(v) \leftarrow 1$ for each vertex $v$ of $G$ 

\State Set the dynamic moment\_of\_regret and rollback\_point values.

\State Sort the edges of $G$ in a non-conflicting order.
\State $c\gets 1$;
\State $i\gets 0$;

\State {\bf do}

\State \hspace{1cm} $C_i \gets C$
\State \hspace{1cm} $i\gets i+1$;
\State \hspace{1cm} add edge $e_i=uv$ to $G$;

\State \hspace{1cm} \textbf{if} $C(u) = C(v)$ \textbf{then} %\Comment edge $e_i$ creates conflict
\State \hspace{1.5cm} \textbf{if} $c' \leq c$ is not assigned to a neighbor of $u$ \textbf{then}
\State \hspace{2cm} $C(u) \gets c'$;
\State \hspace{1.5cm} \textbf{else if}
$c' \leq c$ is not assigned to a neighbor of $v$ \textbf{then}
\State \hspace{2cm} $C(v) \gets c'$;
%\comment edge $e_i$ creates conflict 
\State \hspace{1.5cm} \textbf{else}
\State \hspace{2cm} $c\gets c+1;$
\State \hspace{2cm} $C(u) = c$; \Comment add minimum possible color;

%\State \hspace{1cm} \textbf{else}
%\State \hspace{1.5cm} choose minimum available color;
\State \hspace{1cm} \textbf{end if}
\State \hspace{1cm} \textbf{if} is\_dyn\_moment\_of\_regret \textbf{then}
\State \hspace{1.5cm} $j = i$ $-$ dyn\_rollback\_value
\State \hspace{1.5cm} $G' \gets G - \{e_{j+1}, \cdots e_i\}$;
%[dyn\_rollback\_point edges];
\State \hspace{1.5cm} $G'' \gets $ \textbf{DGC\_FPT}($G',G,C_j,i,c-1-|C_j|$)
\item \hspace{2cm} \textbf{if} c(G'') \textless \space c(G) \textbf{then}
\State \hspace{2.5cm} $G \gets G'' $;
\State \hspace{2.5cm} $c \gets c(G'') $;
\State \hspace{2cm} \textbf{end if}
\State \hspace{1cm} \textbf{end if}
\State \textbf{while} (Not all the edges of the graph are added); 
\State \textbf{Return} $c$ as the final number of colors assigned to $G$;
\EndProcedure
    \end{algorithmic}
    \caption{Dynamic Turbo-charged Heuristic for Graph Coloring}
    \label{alg:pmm}
\end{algorithm}

\vspace{10pt}
 
We observed that, after each moment of regret, %that the greedy approach leads to, 
the FPT subroutine often decreases the coloring by at least one color. %Therefore, turbo-charging the greedy approach could enhance the solution at each regret point reached.
%Once a new better solution is found, it replaces the previous one and the algorithm will pursue with the edge addition.
As shown in the previous section, the running time of the subroutine is in $O^*(4^k)$, but in our case $k$ could be small and the number of times we call the FPT subroutine is bounded above by (a fraction of) the number of colors.
Despite the additional processing time, 
%Even thought the algorithm will be adding more processing time to the greedy approach when we reach the moment of regret, knowing that 
we almost always have a better coloring, which is the main objective of the turbo-charging method. 
%Moreover, the time added is only couple of seconds since we are partially coloring the graph.

In the next section, the reported preliminary experimental results show how the dynamic turbo-charging method, applied to the above simple greedy heuristic, was consistently effective by enhancing all the known results obtained using heuristic methods.
%in a small amount of time.

\section{Experimental Analysis}

%Our algorithm was implemented in the Java programming language and tested on a machine of the Windows 10 operating system with a CPU of 2.5 GHz, Intel Core i7, and 4 GB memory.  

Our algorithm was tested on several benchmarks from the DIMACS graph coloring instances. Due to the successive applications of our fixed-parameter (turbo-charging) subroutine, the algorithm is slower than the usual efficient heuristic methods but it takes only a few seconds in general. 
We compare the results (number of colors) obtained by our algorithm to those obtained by the most known heuristics. 
For this purpose, a subset of the DIMACS graph coloring benchmarks was chosen to maximize the overlap with the published
results of heuristic methods.
The obtained results are reported in the tables below. The values for the chromatic number are placed in parentheses when they are the lowest values reported in the literature but not proven to be the minimum.

Table \ref{tab:tcRcc} below provides a comparison between the results obtained by our turbo-charged heuristic, referred to by  ``DYN-TC,'' and the Range Compaction Coloring heuristic (RCC) from \cite{rangecomp2003}.
As the table shows, our turbo-charged heuristic either matches or obtains better results than the RCC method on all the test cases.

%\vspace{2 mm}
\begin{table} [htb!]
\parbox{.85\linewidth}{
\centering
\begin{tabular}
{|c|c|c|c|}
\hline
\multirow{2}{*}{Graph $G$} & \multirow{2}{*}{$\chi(G)$} & \multirow{2}{*}{DYN-TC} & %
    \multirow{2}{*}{RCC} \\
 & & & \\
\hline
 DSJC125.5 & (17) & 19 & 20 \\
\hline
 DSJC250.5 & (28) & 30 & 33 \\
 \hline
 DSJC500.5 & (48) & 55 & 56 \\
 \hline
 DSJC1000.5 & (84) & 93 & 98 \\
 \hline
 DSJR500.1 & 12 & 12 & 12 \\
 \hline
 DSJR500.5 & (123) & 126 & 131 \\
 \hline
 R125.1 & 5 & 5 & 5 \\
 \hline
 R125.5 & 36 & 36 & 38 \\
 \hline
 R250.1 & 8 & 8 & 8 \\
 \hline
 R250.1c & 64 & 64 & 64 \\
 \hline
 Flat300\_20\_0 & 20 & 24 & 21 \\
 \hline
 Flat300\_26\_0 & 26 & 31 & 37 \\
 \hline
 Flat300\_28\_0 & 28 & 33 & 37 \\
 \hline
 Flat1000\_50\_0 & 50 & 92 & 96 \\
 \hline
 Flat1000\_60\_0 & 60 & 92 & 96 \\
 \hline
 Flat1000\_76\_0 & 76 & 93 & 97 \\
\hline
% etc. ...
\end{tabular}
\caption{Dynamic Turbo-charging versus RCC.}
\label{tab:tcRcc}
}
\end{table}

Table \ref{tab:tcGreedy} shows the comparison between our approach and the (folklore) greedy heuristic. The reported results show tremendous improvements in terms of number of colors obtained. 

\begin{table} [H]
\parbox{.85\linewidth}{
\centering
\begin{tabular}
{|c|c|c|c|}
\hline
\multirow{2}{*}{Graph $G$} & \multirow{2}{*}{$\chi(G)$} & \multirow{2}{*}{DYN-TC} & %
    \multirow{2}{*}{Greedy} \\
 & & & \\
\hline
 DSJC125.5 & (17) & 19 & 26 \\
\hline
 DSJC250.5 & (28) & 30 & 43 \\
 \hline
 DSJC500.5 & (48) & 55 & 72 \\
 \hline
 DSJC1000.5 & (84) & 93 & 127 \\
 \hline
 DSJR500.1 & 12 & 12 & 15 \\
 \hline
 DSJR500.5 & (123) & 126 & 143 \\
 \hline
 R125.1 & 5 & 5 & 5 \\
 \hline
 R125.5 & 36 & 36 & 44 \\
 \hline
 R250.1 & 8 & 8 & 9 \\
 \hline
 R250.1c & 64 & 64 & 76 \\
 \hline
 Flat300\_20\_0 & 20 & 24 & 47 \\
 \hline
 Flat300\_26\_0 & 26 & 31 & 45 \\
 \hline
 Flat300\_28\_0 & 28 & 33 & 46 \\
 \hline
 Flat1000\_50\_0 & 50 & 92 & 126 \\
 \hline
 Flat1000\_60\_0 & 60 & 92 & 125 \\
 \hline
 Flat1000\_76\_0 & 76 & 93 & 122 \\
\hline
% etc. ...
\end{tabular}
\caption{Dynamic Turbo-charging versus Greedy.}
\label{tab:tcGreedy}
}
\end{table}

%In the following tables, ``DYN-TC'' refers to our dynamically turbo-charged heuristic approach while ``RCC'' is the Range Compaction Coloring heuristic from \cite{rangecomp2003} and ``Greedy'' corresponds to the (folklore) greedy algorithm, also known as the ``sequential'' algorithm. We also added results obtained by "TABU-S" which is the Tabu search Meta-Heuristic algorithm from \cite{HdW87}, modified and improved in \cite{CL96}.

%\vspace{2 mm}

Table \ref{tab:tcTabu} presents a comparison between our approach and the Tabu-Search meta-heuristic, referred to by ``TABU-S,'' introduced in \cite{HdW87} and later improved in \cite{CL96}.
Despite the fundamental difference between the two methods, especially because of the fact that turbo-charging does not seem applicable to meta-heuristics, our approach proves to be competitive, often reducing the resulting number of colors. Moreover, while we could not perform a
comparison between running times, our turbo-charging method showed a smooth increase in running time when the graph size increases, while the TABU-S’s reported running time increase in a more abrupt manner as the graph size gets bigger.

\begin{table} [H]
\parbox{.85\linewidth}{
\centering
\begin{tabular}
{|c|c|c|c|}
\hline
\multirow{2}{*}{Graph $G$} & \multirow{2}{*}{$\chi(G)$} & \multirow{2}{*}{DYN-TC} & %
    \multirow{2}{*}{TABU-S} \\
 & & & \\
\hline
 DSJC125.5 & (17) & 19 & 19 \\
\hline
 DSJC250.5 & (28) & 30 & 31 \\
 \hline
 DSJC500.5 & (48) & 55 & 53 \\
 \hline
 DSJC1000.5 & (84) & 93 & 93 \\
 \hline
 DSJR500.1 & 12 & 12 & 12 \\
 \hline
 DSJR500.5 & (123) & 126 & 126 \\
 \hline
 R125.1 & 5 & 5 & 5 \\
 \hline
 R125.5 & 36 & 36 & 39 \\
 \hline
 R250.1 & 8 & 8 & 8 \\
 \hline
 R250.1c & 64 & 64 & 65 \\
 \hline
 Flat300\_20\_0 & 20 & 24 & 26 \\
 \hline
 Flat300\_26\_0 & 26 & 31 & 34 \\
 \hline
 Flat300\_28\_0 & 28 & 33 & 33 \\
 \hline
 Flat1000\_50\_0 & 50 & 92 & 90 \\
 \hline
 Flat1000\_60\_0 & 60 & 92 & 91 \\
 \hline
 Flat1000\_76\_0 & 76 & 93 & 91 \\
\hline
% etc. ...
\end{tabular}
\caption{Dynamic Turbo-charging versus TABU Search.}
\label{tab:tcTabu}
}
\end{table}

Knowing that the previously stated heuristics do not tackle many instances, we proceed to comparing the results of our algorithm against the optimal solutions especially on DIMACS graphs where we could not find published results obtained by previous heuristics. Table \ref{tab:tcBest} below shows the results of the comparison. Obviously, the number of colors obtained by our algorithm are always close or equal to the best known solutions. 

\begin{table} [H]
\parbox{.85\linewidth}{
\centering
\begin{tabular}
{|c|c|c|}
\hline
\multirow{2}{*}{Graph $G$} & \multirow{2}{*}{$\chi(G)$} & \multirow{2}{*}{DYN-TC} \\
 & & \\
\hline
 DSJC1000.1 & (20) & 20 \\
\hline
 LE450\_25C & 25 & 25 \\
 \hline
 LE450\_25D & 25 & 25 \\
 \hline
 R1000.1C & (98) & 98 \\
 \hline
 C2000.5 & (146) & 150 \\
 \hline
 C4000.5 & (260) & 266 \\
 \hline
 queen11\_11.col & (11) & 11 \\
 \hline
 queen12\_12.col & (?) & 14 \\
 \hline
 queen13\_13.col & (13) & 15 \\
 \hline
 zeroin.i.1.col & (49) & 49 \\
 \hline
 zeroin.i.2.col & (30) & 33 \\
 \hline
 zeroin.i.3.col & (30) & 33 \\
 \hline
 miles250.col & 8 & 8 \\
 \hline
 miles500.col & 20 & 21 \\
 \hline
 miles1000.col & 42 & 45 \\
 \hline
 miles1500.col & 73 & 75 \\
\hline
% etc. ...
\end{tabular}
\caption{Dynamic Turbo-charging versus the stated best-known solutions reported at DIMACS.}
\label{tab:tcBest}
}
\end{table}

The results in the above tables show that, on average, our heuristic improves the current minimum known estimate of the chromatic number of almost all input instances when compared to other heuristics. It is interesting to note that, despite the iterative application of our dynamic fixed-parameter algorithm, the overall running time of our algorithm is increasing in a conditioned and proportional manner depending on the number of vertices (unlike TABU-S for example, which exhibits a sudden and somehow unexpected increase in running time). As mentioned above, our running times are higher since we are employing a fixed-parameter algorithm at certain time intervals. Despite this fact, our algorithm was computing a better result on most of the test cases and sometimes we obtain close-to-optimal (if not optimal) results with better running times. A notable result is the new best-known solution we report in Table \ref{tab:tcBest} for the DIMACS instance named queen12\_12.col.

Finally, we note again that a similar approach that inspired the work on dynamic problems was presented in \cite{HartungN13} for the List Coloring problem. In fact, an enhanced heuristic algorithm for Graph Coloring was presented and a few experiments were reported. We conducted experiments on the same set of graphs used in \cite{HartungN13}. While our Turbo-charging algorithm is notably slower due to the repetitive use of our FPT subroutine at each moment-of-regret, it can sometimes reduce the number of colors further, as shown in the below table.

\begin{table} [htb!]
\parbox{.85\linewidth}{
\centering
\begin{tabular}
{|c|c|c|}
\hline
{Graph $G$} & {Search-Tree $k$} & {DYN-TC $k$} \\
\hline
 ash608GPIA & 5 & 5 \\
\hline
 DSJC1000.1 & 25 & 20 \\
 \hline
 DSJC500.1 & 15 & 15 \\
 \hline
 latin\_square\_10 & 117 & 116 \\
 \hline
 le450\_15a & 16 & 16 \\
 \hline
 qg.order40 & 41 & 41 \\
 \hline
 queen16\_16 & 19 & 19 \\
 \hline
 school1\_nsh & 23 & 21 \\
 \hline
 wap03 & 50 & 50 \\
\hline
% etc. ...
\end{tabular}
\caption{Dynamic Turbo-charging versus the incremental list-coloring method.}
\label{tab:tcBest}
}
\end{table}

\section{Conclusion}

Capitalizing on recent work on turbo-charging heuristics, this paper has presented a novel turbo-charged heuristic for graph coloring that is based on proving that Dynamic Graph Coloring is fixed-parameter tractable with respect to the edit parameter.
%Our turbo-charging method works by using the corresponding fixed-parameter algorithm. 
We evaluated our turbo-charged heuristic on a broad range of DIMACS benchmark graphs and compared results with the most known heuristics on a selected number of instances. Preliminary results reported in this paper show a consistent improvement over the greedy heuristics in addition to occasional improvement over the most known meta-heuristic.
%At this stage, our code is a proof-of-concept implementation that we believe can be enhanced considerably. Moreover, the running time can be further improved if an improved fixed-parameter algorithm for DGC is found and used as turbo-charging subroutine.

We have introduced an enhanced dynamic turbo-charging approach that consists of varying the moment of regret threshold based on problem specific constraints. It would be interesting to study the effect of this new approach on other turbo-charged heuristics such as those studied in \cite{Abu-KhzamCESW17} as well as other problems for which an FPT dynamic version exists, such as those studied in 
\cite{downey2014dynamic}, \cite{abu14}, \cite{abu2015parameterized} and \cite{Krithika2018}.
%
%Recently, a dynamic version of Red Blue Dominating Set was studied in \cite{abu19}. This implicitly extends to dynamic versions of Set Cover and Hitting Set. It would also be interesting to study the connected version of the problem (introduced in \cite{abu11}) and to apply the turbo-charging method to the corresponding domination/covering problems.  

Finally, we have studied for the first time a parameterized dynamic version of a graph partitioning problem. Previous work has focused mainly on problems whose solutions are subsets of the respective sets of vertices in given input graphs. Our fixed-parameter algorithm for Dynamic Graph Coloring can be easily adapted to solve Dynamic Clique Cover and to obtain a turbo-charged heuristic for the problem. It would be interesting to study the application of this turbo-charging method to similar graph partitioning problems. In fact, a dynamic version of the Cluster Editing problem (CE) was recently introduced and studied in \cite{LMNN2018}. It would also be interesting to extend this work to the multi-parameterized version introduced in \cite{abu17} and to possibly apply our enhanced turbo-charging method to such graph partitioning problems.

\bibliographystyle{abbrv}
\bibliography{references}

\begin{thebibliography}{10}

\bibitem{abu17}
F.~N. Abu{-}Khzam.
\newblock On the complexity of multi-parameterized cluster editing.
\newblock {\em J. Discrete Algorithms}, 45:26--34, 2017.

\bibitem{Abu-KhzamCESW17}
F.~N. Abu{-}Khzam, S.~Cai, J.~Egan, P.~Shaw, and K.~Wang.
\newblock Turbo-charging dominating set with an {FPT} subroutine: Further
  improvements and experimental analysis.
\newblock In {\em Theory and Applications of Models of Computation - 14th
  Annual Conference, {TAMC} 2017, Bern, Switzerland, April 20-22, 2017,
  Proceedings}, pages 59--70, 2017.

\bibitem{abu14}
F.~N. Abu{-}Khzam, J.~Egan, M.~R. Fellows, F.~A. Rosamond, and P.~Shaw.
\newblock On the parameterized complexity of dynamic problems with connectivity
  constraints.
\newblock In Z.~Zhang, L.~Wu, W.~Xu, and D.~Du, editors, {\em Combinatorial
  Optimization and Applications - 8th International Conference, {COCOA} 2014,
  Wailea, Maui, HI, USA, December 19-21, 2014, Proceedings}, volume 8881 of
  {\em Lecture Notes in Computer Science}, pages 625--636. Springer, 2014.

\bibitem{abu2015parameterized}
F.~N. Abu{-}Khzam, J.~Egan, M.~R. Fellows, F.~A. Rosamond, and P.~Shaw.
\newblock On the parameterized complexity of dynamic problems.
\newblock {\em Theor. Comput. Sci.}, 607:426--434, 2015.

\bibitem{rangecomp2003}
R.~H. Andrea Di~Blas, Arun~Jagota.
\newblock A range-compaction heuristic for graph coloring.
\newblock {\em Journal of Heuristics}, 9:489--506, 2003.

\bibitem{fourColoring}
K.~Appel and W.~Haken.
\newblock The solution of the four-color-map problem.
\newblock {\em Scientific American}, 237(4):108--121, 1977.

\bibitem{Bjorklund2009}
A.~Bj{\"{o}}rklund, T.~Husfeldt, and M.~Koivisto.
\newblock Set partitioning via inclusion-exclusion.
\newblock {\em {SIAM} J. Comput.}, 39(2):546--563, 2009.

\bibitem{simAnn}
M.~Chams, A.~Hertz, and D.~de~Werra.
\newblock Some experiments with simulated annealing for coloring graphs.
\newblock {\em European Journal of Operational Research}, 32:260--266, 1987.

\bibitem{Chen10}
J.~Chen, I.~A. Kanj, and G.~Xia.
\newblock Improved upper bounds for vertex cover.
\newblock {\em Theor. Comput. Sci.}, 411(40-42):3736--3756, 2010.

\bibitem{CL96}
J.~C. Culberson and F.~Luo.
\newblock Exploring the k-colorable landscape with iterated greedy.
\newblock pages 245--284, 1995.

\bibitem{downey2014dynamic}
R.~G. Downey, J.~Egan, M.~R. Fellows, F.~A. Rosamond, and P.~Shaw.
\newblock Dynamic dominating set and turbo-charging greedy heuristics.
\newblock {\em Tsinghua Science and Technology}, 19(4):329--337, Aug 2014.

\bibitem{downey1999parameterized}
R.~G. Downey and M.~R. Fellows.
\newblock {\em Parameterized Complexity}.
\newblock Monographs in Computer Science. Springer, 1999.

\bibitem{npHard1979}
M.~R. Garey and D.~S. Johnson.
\newblock {\em Computers and Intractability: A Guide to the Theory of
  NP-Completeness}.
\newblock W.H. Freeman, 1979.

\bibitem{npHard1974}
M.~R. Garey, D.~S. Johnson, and L.~Stockmeyer.
\newblock Some simplified np-complete problems.
\newblock In {\em Proceedings of the Sixth Annual ACM Symposium on Theory of
  Computing}, STOC '74, pages 47--63, New York, NY, USA, 1974. ACM.

\bibitem{HartungN13}
S.~Hartung and R.~Niedermeier.
\newblock Incremental list coloring of graphs, parameterized by conservation.
\newblock {\em Theor. Comput. Sci.}, 494:86--98, 2013.

\bibitem{HdW87}
A.~Hertz and D.~de~Werra.
\newblock Using tabu search techniques for graph coloring.
\newblock {\em Computing}, 39:345--351, 1987.

\bibitem{Krithika2018}
R.~Krithika, A.~Sahu, and P.~Tale.
\newblock Dynamic parameterized problems.
\newblock {\em Algorithmica}, 80(9):2637--2655, Sep 2018.

\bibitem{DBLP:books/sp/Lewis16}
R.~M.~R. Lewis.
\newblock {\em A Guide to Graph Colouring - Algorithms and Applications}.
\newblock Springer, 2016.

\bibitem{LMNN2018}
J.~Luo, H.~Molter, A.~Nichterlein, and R.~Niedermeier.
\newblock Parameterized dynamic cluster editing.
\newblock {\em CoRR}, abs/1810.06625, 2018.

\bibitem{gInter}
D.~W. Matula, G.~Marble, and J.~D. Isaacson.
\newblock Graph coloring algorithms.
\newblock {\em Graph Theory and Computing}, pages 109--122, 1972.

\bibitem{naomi}
N.~Nishimura.
\newblock Introduction to reconfiguration.
\newblock {\em Algorithms}, 11(4):52, 2018.

\bibitem{gLarFst}
D.~J.~A. Welsh and M.~B. Powell.
\newblock An upper bound for the chromatic number of a graph and its
  application to timetabling problems.
\newblock {\em The Computer Journal}, 10:85--86, 1967.

\end{thebibliography}

\end{document}